\documentclass[10pt,journal,a4paper]{IEEEtran}%TWC
\usepackage{amssymb}
\usepackage{amsmath}
\usepackage{bm}
\usepackage{comment}
\usepackage{amsthm}
\usepackage{graphicx}
\usepackage{gensymb}
\usepackage{subfigure}
\usepackage{float}
\usepackage{cite}
\usepackage{enumerate}
\usepackage{enumitem}
\usepackage{multirow}
\usepackage{amsfonts}
\usepackage{url}
\usepackage{diagbox}
\usepackage{color, soul}
\usepackage{amsthm}
\usepackage{type1cm}
\DeclareMathOperator*{\argmax}{arg\,max}

\newcommand{\titlefont}{\fontsize{16.5pt}{24pt}\selectfont}

\begin{document}
\newtheorem{remark}{\bf~~Remark}
\newtheorem{proposition}{\bf~~Proposition}
\newtheorem{theorem}{\bf~~Theorem}
\newtheorem{definition}{\bf~~Definition}
%\font\myfont=cmr12 at 14pt
\title{{\titlefont Intelligent Omni-Surfaces Aided Wireless Communications: Does the Reciprocity Hold?}}

\author{
\IEEEauthorblockN{
    {Shaohua Yue},
	{Shuhao Zeng}, \IEEEmembership{Student Member, IEEE},
	{Hongliang Zhang}, \IEEEmembership{Member, IEEE},\\
	{Fenghan Lin}, \IEEEmembership{Senior Member, IEEE},
	{Liang Liu}, \IEEEmembership{Member, IEEE},
	{and Boya Di}, \IEEEmembership{Member, IEEE}}
	\vspace{-0.9cm}
	%\thanks{Manuscript received Dec. 20, 2020; accepted Feb. 22, 2021. This work was supported in part by the National Natural Science Foundation of China under Grants 61625101, 61829101, and 61941101, and in part by NSF EARS-1839818, CNS-1717454, CNS-1731424, and CNS-1702850. The associate editor coordinating the review of this paper and approving it for publication was C. You. \emph{(Corresponding author: Hongliang Zhang.)}}
	\thanks{S. Yue, S. Zeng, and B. Di are with the Department of Electronics, Peking University, Beijing 100871, China. (email: \{yueshaohua; shuhao.zeng; boya.di\}@pku.edu.cn).} 
	\thanks{Hongliang Zhang is with the Department of Electrical and Computer Engineering, Princeton University, Princton, NJ 08544, USA. (email:hz16@princeton.edu).} 
	\thanks{Fenghan Lin is with the School of Information Science and Technology, ShanghaiTech University, Shanghai  201210, China. (email: linfh@shanghaitech.edu.cn).}
	\thanks{Liang Liu is with the Department of Electronic and Information Engineering, The Hong Kong Polytechnic University, Hong Kong, SAR, China. (e-mail:liang-eie.liu@polyu.edu.hk).}

	%\thanks{H. Zhang and H. V. Poor are with Department of Electrical Engineering, Princeton University, Princeton, NJ 08544 USA (email: hongliang.zhang92@gmail.com, poor@princeton.edu).}
	%\thanks{B. Di is with Department of Computing, Imperial College London, London SW7 2AZ, U.K. (email: diboya92@gmail.com).}
	%\thanks{Z. Han is with Electrical and Computer Engineering Department, University of Houston, Houston, TX 77004 USA, and also with the Department of Computer Science and Engineering, Kyung Hee University, Seoul 02447, South Korea (email: zhan2@uh.edu).}
}

\maketitle

\begin{abstract}
Intelligent omni-surfaces (IOS) have attracted great attention recently due to its potential to achieve full-dimensional  communications by simultaneously reflecting and refracting signals toward both sides of the surface. 
However, it still remains an open question whether the reciprocity holds between the uplink and downlink channels in the IOS-aided wireless communications.
In this work, we first present a physics-compliant IOS-related channel model, based on which the channel reciprocity is investigated. We then demonstrate the angle-dependent electromagnetic response of the IOS element in terms of both incident and departure angles. This serves as the key feature of IOS that drives our analytical results on beam non-reciprocity. Finally, simulation and experimental results are provided to verify our theoretical analyses.
\end{abstract}


\section{Introduction}
Reconfigurable intelligent surfaces (RISs) have been viewed as a promising technology to realize ultra-massive multi-input multi-output (MIMO) for the next-generation communications~\cite{RISintro}. Recently, the intelligent omni-surface (IOS), a novel type of RIS, has drawn considerable attention in both academia and industry~\cite{IOSintro}. As shown in Fig.~\ref{scenario}, unlike traditional reflective-only RIS which reflects incident signals, the IOS can simultaneously reflect and refract signals toward desired users~\cite{IOS}, thereby improving the spectrum efficiency and cell coverage on both sides of the surface to support full-space wireless transmission~\cite{IOScover}.

Existing IOS-related researches mainly focus on beamforming design\cite{IOS_impact} and performance analysis\cite{IOSana}. 
Authors in \cite{IOS_impact} maximized the sum rate of an IOS-aided downlink communication system by jointly optimizing the phase shifts of the IOS and the digital beamforming vector of the base station (BS).
In \cite{IOSana}, the outage probabilities of IOS-aided orthogonal multiple access networks and IOS-aided non-orthogonal multiple access networks were analyzed and compared. 

However, few of these works addressed the reciprocity of the IOS-based uplink and downlink channels. 
Note that the conventional channel reciprocity is based on the reciprocal transmission of electromagnetic (EM) waves through the air. In contrast, the deployment of IOS changes the EM environment, making it unknown whether IOS has the same impact on EM waves when the signal impinges on different sides of the IOS. 
Therefore, it is non-trivial to judge the existence of reciprocity in IOS-enabled communications. This remains an important yet open question which directly influences both the channel estimation and beamforming scheme design.

To answer the above question, two challenges are to be addressed. \emph{First}, to theoretically analyze the reciprocity of the IOS-aided channels, a reasonable and physics-compliant channel model is required for both the uplink and downlink channels. \emph{Second}, beam reciprocity of IOS should be discussed. Specifically, for a given IOS configuration, the transmitter reflects or refracts the signal to a receiver at a certain direction via the IOS. It is unknown whether the beamforming still works after the directions of the transmitter and the receiver devices exchange. If such reciprocity holds, the IOS can have the same configuration for both the uplink and the downlink beamforming. 

To cope with the above challenges, in this letter, we model the downlink and uplink channels based on the propagation principles of EM waves and the operation principles of the IOS. 
%We analyze the proposed channel model and prove that the uplink and downlink channels are equal. 
%Besides, we show that IOS configuration in IOS-aided beamforming is influenced by both the incident angle and the departure angle. 
We answer the open question by pointing out that the channel reciprocity holds whereas the beam reciprocity does not. It is supported by our theoretical finding that the EM response of the IOS element is influenced by both the incident angle and the departure angle. Finally, an IOS-aided communication testbed is built to verify the analytical results of the channel and beam reciprocity as well as the incidence-departure angle-dependency via experiments.
\section{System Model}\label{sys_mod}
%\vspace{-0.2cm}
%\vspace{-0.1cm}
%In this section, we first introduce the system model of RIS-assisted networks. Afterwards, characteristics of the RIS are discussed, and the channel model is presented.
%\vspace{-0.3cm}
\subsection{Scenario Description}
%\vspace{-0.1cm}
\begin{figure}[t]
\centerline{\includegraphics[width=7.5cm,height=6cm]{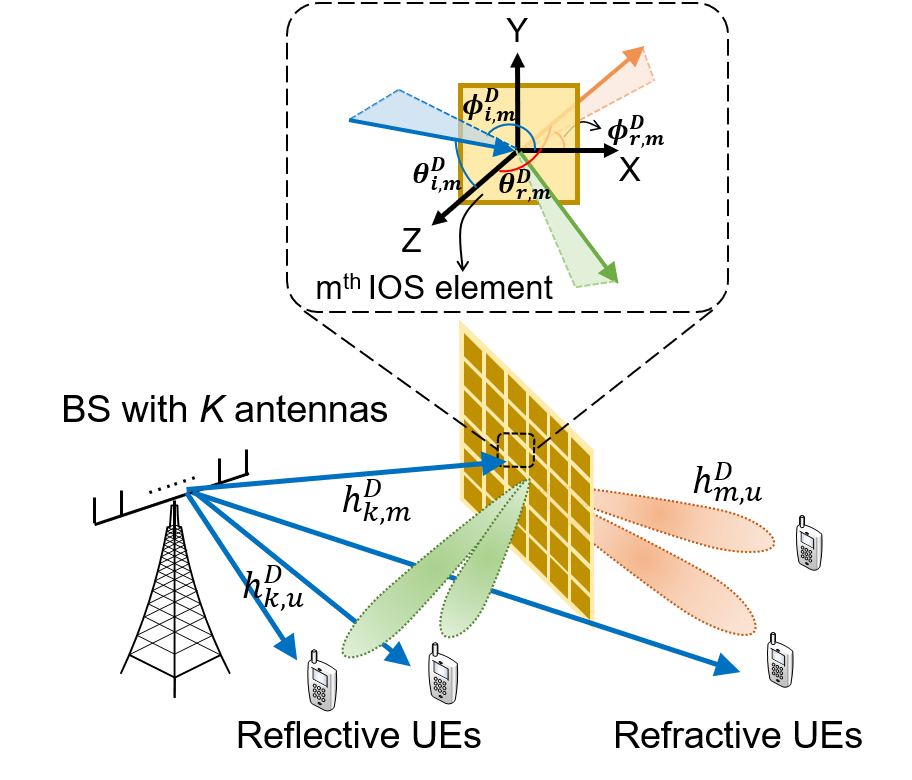}}
\caption{IOS-assisted wireless communication system model.}
\vspace{-0.6cm}
\label{scenario}
\end{figure}

As shown in Fig.~\ref{scenario}, we consider a multi-user MIMO network, where one BS of $K$ antennas communicates with multiple single-antenna users. 
%Due to the dynamic wireless environment with unexpected fading, the direct links between the users and the BS can be unreliable or possibly go down completely. To solve this issue, t
The IOS is deployed between the BS and users to reflect or refract the transmit signal toward the users so that the link quality can be enhanced. With the help of the IOS, signals can be reflected or refracted toward the targets.

\begin{comment}
{\color{blue}The Cartesian coordinates are introduced to demonstrate the communication system, where the IOS is placed at the $x-y$ plane and the origin is located in the center of the IOS. The IOS has $M$ elements and is deployed to assist the communication between the BS and the users. The size of each element of the IOS is $d_x, d_y$ and $d_z$ along the x axis, y axis and z axis, respectively. The region $z < 0$ is referred to as the transmission side of the IOS and the region  $z > 0$  is referred to as reflection side of the IOS. The set of users that are located at the reflection side and transmission side are represented by set $U_R$ and $U_T$, respectively. Therefore, we have $U = U_T \cup U_R$.}
{\color{red}xxxxxxxxxxxxxxxxxxxxxxxxxxxxxxxxxxxxxxxxxxx}
are located in the far field of each IOS element\footnote{The following model and analyses are compatible if the BS is equipped with array antenna because of the far field condition.}. 
\end{comment}
  
\vspace{-0.3cm}
\subsection{Channel Model}
For convenience, we use $\mathcal{X} \in \{D, U\}$ as an indicator to represent the downlink or uplink channel.
The channel between the $k^{th}$ antenna of the BS and the antenna of user $u$ consists of the direct link $h_{k,u}^{\mathcal{X}}$ and $M$ IOS-based scattering channels, where the $m^{th}$ scattering channel 
%represents the channel from antenna $k$ to user $u$ via the $m^{th}$ IOS element. The $m^{th}$ IOS-based scattering channel 
is the cascade of channel $\tilde{h}_{k,m}^{\mathcal{X}}$ between the $k^{th}$ antenna of the BS and the $m^{th}$ IOS element, the EM response $g_{m}^{\mathcal{X}}$ of the $m^{th}$ IOS element to the incident signal\footnote{The EM response $g_{m}^{\mathcal{X}}$ induces a change in the amplitude and phase of the incident signal.}, and the channel $\hat{h}_{m, u}^{\mathcal{X}}$ between the $m^{th}$ IOS element and the user $u$\cite{star_ches}, respectively. Depending on whether the $m^{th}$ IOS element reflects or refracts the signal,  $g_{m}^{\mathcal{X}}$ takes different values as follows

\begin{equation}
g_{m}^{\mathcal{X}} = \begin{cases}
g_{m}^{\mathcal{X},R}, & {\rm signal \ is \ reflected,} \\
g_{m}^{\mathcal{X},T}, & {\rm signal \ is \ refracted.}\\
\end{cases}\label{allo_edge_limit}
\end{equation}

Therefore, the effective channel $H^{\mathcal{X}}_{k,u}$ between the $k^{th}$ antenna of the BS and the user $u$, which takes the IOS-related channels into account, is expressed as

\begin{equation}
    H^{\mathcal{X}}_{k,u} = h_{k,u}^{\mathcal{X}} + \sum_{m} \tilde{h}_{k,m}^{\mathcal{X}}g_{m}^{\mathcal{X}} \hat{h}_{m,u}^{\mathcal{X}},\label{channel}
\end{equation}
where $h_{k,u}^{\mathcal{X}}$ represents the direct link between the $k^{th}$ antenna of the BS and the user antenna. $\tilde{h}_{k,m}^{\mathcal{X}}$ and $\hat{h}_{m,u}^{\mathcal{X}}$ are the channel between the $k^{th}$ BS antenna and the $m^{th}$ IOS element and the channel between the $m^{th}$ IOS element and the user antenna, respectively, which are modeled as\cite{tang}

\begin{subequations}
\begin{align}
    &\tilde{h}_{k,m}^{\mathcal{X}} = \frac{\sqrt{G_{k}F_{k}(\theta_{i/r,m}^{\mathcal{X}},\phi_{i/r,m}^{\mathcal{X}})}}{\sqrt{4\pi}d_{k,m}} e^{-j\frac{2\pi d_{k,m}}{\lambda}},\label{b2i}\\
    &\hat{h}_{m,u}^{\mathcal{X}} = \frac{\lambda\sqrt{G_{u}F_{u}(\theta_{i/r,m}^{\mathcal{X}},\phi_{i/r,m}^{\mathcal{X}})}}{4\pi d_{u,m}}e^{-j\frac{2\pi d_{u,m}}{\lambda}},\label{i2u}
\end{align}
\end{subequations}
where $G_k$ and $G_u$ are the antenna gain of the $k^{th}$ BS antenna and user antenna;
$F_k$ and $F_u$ are the radiation pattern of the $k^{th}$ BS antenna and user antenna; $d_{k,m}$ is the distance between the $k^{th}$ antenna of the BS and the $m^{th}$ IOS element; $d_{u,m}$ is the distance between the antenna of user $u$ and the $m^{th}$ IOS element; $\lambda$ is the signal wavelength. The notation $i/r$ represents the incident signal from the antenna to the IOS element or the signal that are reflected or refracted from the IOS element to the antenna. For instance, as shown in Fig.~\ref{scenario}, $(\theta_{i,m}^{D}, \phi_{i,m}^{D})$ represent the elevation angle and the azimuth angle from the $k^{th}$ antenna of the BS to the $m^{th}$ IOS element in the downlink channel and $(\theta_{r,m}^{D}, \phi_{r,m}^{D})$ represent the elevation angle and the azimuth angle from the $m^{th}$ IOS element to the user antenna in the downlink channel. The elevation angle and the azimuth angle can be defined similarly for the uplink channel.

% $F_{k}(\theta_{i/r,m}^{\mathcal{X}},\phi_{i/r,m}^{\mathcal{X}})$ is the radiation pattern of the $k^{th}$ BS antenna pointing at the $m^{th}$ IOS element; $F_{u}(\theta_{i/r,m}^{\mathcal{X}} \phi_{i/r,m}^{\mathcal{X}})$ is the radiation pattern whose direction is from the antenna at user $u$ to the $m^{th}$ IOS element. 

\begin{comment}
\begin{proof}
    See Appendix~\ref{app_channelmodel}.
\end{proof}
\end{comment}

The EM response of the $m^{th}$ IOS element to the incident signal, i.e., $g_{m}^{\mathcal{X}}$, is an angle-dependent value and modeled as~\cite{IOS_impact}
\begin{equation}
    g_{m}^{\mathcal{X}}=\sqrt{G_mF_m(\theta_{i,m}^{\mathcal{X}},\phi_{i,m}^{\mathcal{X}})F_m(\theta_{r,m}^{\mathcal{X}},\phi_{r,m}^{\mathcal{X}})S_m}\Gamma_m,\label{impact}
\end{equation}
where $G_m$ is the power gain, $F_m$ is the radiation pattern, and $S_m$ is the surface area of the IOS element, which are same for all IOS element, i.e., $G_m=G_{m'}, F_m=F_{m'}, S_m=S_{m'}, \forall m'\in \{1,2...M\}$.  $(\theta_{i,m}^{\mathcal{X}},\phi_{i,m}^{\mathcal{X}})$ is the incident angle from the transmitting antenna to the $m^{th}$ element, $(\theta_{r,m}^{\mathcal{X}},\phi_{r,m}^{\mathcal{X}})$ is the departure angle from the $m^{th}$ element to the receiving antenna. Moreover, the reflection/refraction coefficient of the $m^{th}$ IOS element is denoted as $\Gamma_m=\beta_m e^{j\psi_m}$, where $\beta_m$ and $\psi_m$ are the amplitude and phase shift adjustments that are imposed on the signal. $\Gamma_m$ is affected by the IOS element state, the signal incident angle and the departure angle. 
%\subsubsection{Uplink Channel}
\begin{comment}
Similarly, the uplink channel from user $u$ to the $k^{th}$ antenna of the BS consists of the direct link $h_{u,k}^{U}$ and $M$ IOS-based scattering channels, the $m^{th}$ of which is the cascade of the channel $h_{u,m}^{U}$ from user $u$ to the $m^{th}$ IOS element,  $g^{U}_{u,m,k}$, the EM response of the IOS element to the uplink incident signal, and $h_{m,k}^{U}$, the channel from the $m^{th}$ IOS element to the $k^{th}$ antenna of the BS. Therefore, the uplink channel $h^{UL}$ can be expressed as
\begin{equation}
     h^{UL} = h_{u,k}^{U} + \sum_{m}h_{u,m}^{U}g_{u,m,k}^{U}h_{m,k}^{U},\label{uplink}   
\end{equation}
where channels $h_{u,m}^{U}$, $h_{u,k}^{U}$ and $h_{m,k}^{U}$ are modeled as
\begin{subequations}
\begin{align}
    &h_{u,k}^{U} =  \frac{\lambda\sqrt{G_{u}F_{u}(\theta_{i},\phi_{i})G_{k}F_{k}(\theta_{r},\phi_{r})}}{4\pi d_{u,k}} e^{-j\frac{2\pi d_{u,k}}{\lambda}},
\label{directuplink}\\
   &h_{u,m}^{U} = \frac{\sqrt{G_{u}F_{u}(\theta_{i,m},\phi_{i,m})}}{\sqrt{4\pi}d_{u,m}} e^{-j\frac{2\pi d_{u,m}}{\lambda}},\label{u2i} \\
   &h_{m,k}^{U} = \frac{\lambda\sqrt{G_{k}F_{k}(\theta_{r,m},\phi_{r,m})}}{4\pi d_{k,m}}e^{-j\frac{2\pi d_{k,m}}{\lambda}}.\label{i2b}
    \end{align}
\end{subequations}
\end{comment}

\section{Theoretical Analysis on Channel Model}
Based on the proposed IOS-based channel model, we give the following analyses on the channel reciprocity. Then we discuss the impact of IOS angle dependency on beamforming and channel estimation.
\vspace{-10pt}
\subsection{Analysis on the Channel Reciprocity}

%Different from traditional RIS which only reflects signal, the IOS both reflect and refract signal, bringing different impacts on the reflected and refracted signals. However, in the following 

%{\color{red}Therefore, channel reciprocity for reflected signal and refracted signal should be considered independently for IOS. The reflection and refraction channels have correlations, because the reflection coefficient and transmission coefficient are both determined by the signal direction and element configuration. Besides, while the assumption of ideal power reflection is adopted in most RIS-based research, power splitting between reflected signal and refracted signal should be considered for IOS for accurate channel estimation and beamforming.}

\begin{proposition}\label{proposition1}
The channel reciprocity holds in the IOS-aided system. In detail, the cascaded uplink channel user-IOS and IOS-BE, i.e., $\hat{h}^{U}_{m,u}\tilde{h}^{U}_{k,m}$, equals the cascaded downlink channel BS-IOS and IOS-user, i.e., $\tilde{h}^{D}_{k,m}\hat{h}^{D}_{m,u}$.
\end{proposition}

\begin{proof}
Based on the definition of  $\tilde{h}^{\mathcal{X}}_{k,m},\hat{h}^{\mathcal{X}}_{m,u}$ in (\ref{b2i}), (\ref{i2u}),
\begin{equation}
\begin{aligned}
    &\hat{h}^{U}_{m,u}\tilde{h}^{U}_{k,m}\\
    &=\frac{\sqrt{G_{u}F_{u}(\theta_{i,m}^U,\phi_{i,m}^U)G_{k}F_{k}(\theta_{r,m}^U,\phi_{r,m}^U)}}{{4\pi}^{1.5}d_{u,m}d_{k,m}} e^{-j\frac{2\pi (d_{u,m}+d_{k,m})}{\lambda}}.
\end{aligned}
\end{equation}
\begin{equation}
\begin{aligned}
    &\tilde{h}^{D}_{k,m}\hat{h}^{D}_{m,u}\\
    &=\frac{\sqrt{G_{k}F_{k}(\theta_{i,m}^D,\phi_{i,m}^D)G_{u}F_{u}(\theta_{r,m}^D,\phi_{r,m}^D)}}{{4\pi}^{1.5}d_{k,m}d_{u,m}} e^{-j\frac{2\pi (d_{k,m}+d_{u,m})}{\lambda}},
\end{aligned}
\end{equation}

For the uplink channel and the downlink channel, the signal direction at the transmitter and the receiver keeps the same, i.e., $(\theta_{r,m}^{U},\phi_{r,m}^{U}) = (\theta_{i,m}^{D},\phi_{i,m}^{D})$ and $(\theta_{i,m}^{U},\phi_{i,m}^{U}) = (\theta_{r,m}^{D},\phi_{r,m}^{D})$. Therefore, $\hat{h}^{U}_{m,u}\tilde{h}^{U}_{k,m} = \tilde{h}^{D}_{k,m}\hat{h}^{D}_{m,u}$.
\end{proof}

\begin{proposition}\label{proposition2}
As shown in (\ref{channel}), the downlink channel $H^{D}_{k,u}$ and uplink channel $H^{U}_{k,u}$ of IOS-aided communication system are equal.
\end{proposition}
\begin{proof}
According to the channel model in (\ref{channel}), we have
\begin{subequations}
\begin{align}
H^{D}_{k,u} = h_{k,u}^{D} + \sum_{m}\tilde{h}_{k,m}^{D}g_{m}^{D}\hat{h}_{m,u}^{D},\label{downlink}\\
H^{U}_{k,u} = h_{k,u}^{U} + \sum_{m}\hat{h}_{m,u}^{U}g_{m}^{U}\tilde{h}_{k,m}^{U}.\label{uplink}
\end{align}
\end{subequations}
Firstly, the direct links satisfy channel reciprocity, i.e., $h_{k,u}^D=h_{k,u}^U$. Secondly, IOS's EM response to the signal is reciprocal, i.e., $g_{m}^D=g_{m}^U$, according to (\ref{impact}), since the reflection/refraction coefficient $\Gamma$ is the same  for both uplink and downlink paths. This is guaranteed by the fact that IOS does not contain time-variant or nonlinear material such that the Rayleigh-Carson reciprocity theorem \cite{gmreciprocity} holds. Thirdly, according to Proposition ~\ref{proposition1}, we have $\tilde{h}^{D}_{k,m}\hat{h}^{D}_{m,u}=\hat{h}^{U}_{m,u}\tilde{h}^{U}_{k,m}$. Therefore, the uplink channel equals the downlink channel, i.e., $H^{D}_{k,u}=H^{U}_{k,u}$.

\end{proof}

\begin{comment}
\begin{remark}
The IOS 's EM response to the signal $g_{m}^{\mathcal{X}}$ is affected by the incident angle of the signal. \label{proposition1}
\end{remark}
According to \cite{IOS}, the IOS element can be modeled with an equivalent circuit whose parameters depend on the incident angle. Since $\Gamma$ is the function of the impedance of the free space and circuit parameters of the IOS element, the variance of incident angle will change $\Gamma$ of the element, thereby affecting $g_{m}^{\mathcal{X}}$.
\end{comment}
This proposition is also verified by the experiment given in section IV-D.
\subsection{Influence of the IOS Angle Dependency on Beamforming}

\begin{remark}\label{remarkbeam}
The IOS element's EM response, i.e., $g_m^{\mathcal{X}}$ in (\ref{impact}), is angle-dependent and thus, the ideal phase shift assumption does not hold. Such angle-dependent property should be considered for accurate IOS-aided beamforming.
\end{remark}

In more details, under the ideal phase shift assumption,  the set of phase shifts is ideally set as $\{0, \pi\}$ for one-bit coded RIS~\cite{idealphase}. However, the  phase shift  is not constant and is affected by both the incident angle and the departure angle in practice. \emph{First}, one may observe the influence of the incident angle on the phase shift $\psi_m$ based on the full-wave simulation as will be shown in Section IV-A. \emph{Second}, given the reciprocity of $\Gamma$,  $\psi_m$ is also influenced by the departure angle. If $\psi_m$ is not influenced by the departure angle, the uplink and downlink channels for the same path can have different incident angles toward the IOS element. These two channels are then unequal, which contradicts with the Proposition~\ref{proposition2}. Thus, we can conclude that the phase shift $\psi_m$ is influenced by both incident and departure angle.  

Since the ideal phase shift assumption does not hold, beamforming methods based on such an assumption may suffer from performance degradation, which will be illustrated by Fig.~\ref{tran_beamform} and Fig.~\ref{refl_beamform} in Section IV-B.

Though the downlink and uplink channels of IOS are equal, the angle dependency of IOS leads to the beam non-reciprocity for the IOS, which is depicted as below. 

\begin{definition}
Denote ($\theta_0,\phi_0$) as the incident angle of the signal impinging on the IOS and ($\theta_1,\phi_1$) as the main beam direction of the reflected/refracted signal. Denote ($\theta_2,\phi_2$) as the main beam direction of the reflected/refracted signal if the incident signal is at direction ($\theta_1,\phi_1$). The IOS configuration is the same for the above two IOS-assisted beamforming.

\emph{IOS satisfies the} beam reciprocity \emph{if:} $(\theta_0,\phi_0) = (\theta_2,\phi_2)$, $\forall (\theta_0,\phi_0)$. Otherwise, IOS is beam non-reciprocal.
\end{definition}

\begin{proposition}
Beam reciprocity does not hold for IOS.\label{beamremark}
\end{proposition}

\begin{proof}
We give the proof by contradiction. Consider a far-field communication scenario where the angle of incidence and angle of departure to all IOS elements are the same. For instance, $\theta_{0,m} \approx \theta_{0}$. The downlink channel refers to the link between the transmitter with an incident angle $(\theta_0, \phi_0)$ and a destination at direction $(\theta_1, \phi_1)$, which is denoted as

\vspace{-15pt}
\begin{equation}
\begin{aligned}
    &(\theta_1, \phi_1)= \\
    &\argmax_{\theta,\phi}(h_{k,u}^{D}+ \sum_{m}\tilde{h}_{k,m}^{D}(\theta_0,\phi_0)g_{m}^D(\theta_0,\phi_0,\theta,\phi)\hat{h}_{m,u}^{D}(\theta,\phi)).
    \end{aligned}
    \label{appB:downbeam}
\end{equation}
The uplink channel refers to the link between the transmitter with an incident angle $(\theta_1, \phi_1)$ and a destination at direction $(\theta_2, \phi_2)$. If the beam reciprocity holds, such uplink channel is denoted as
 \begin{equation}
\begin{aligned}
    &(\theta_0, \phi_0)=(\theta_2, \phi_2)=\\&\argmax_{\theta,\phi}(h_{u,k}^{U}+  \sum_{m}\hat{h}_{u,m}^{U}(\theta_1,\phi_1)g_{m}^U(\theta_1,\phi_1,\theta,\phi)\tilde{h}_{m,k}^{U}(\theta,\phi)).
    \end{aligned}
\end{equation}
According to proposition~\ref{proposition1} and proposition~\ref{proposition2}, the above equation is transformed as
 \begin{equation}
\begin{aligned}
    &(\theta_0, \phi_0)=\\&\argmax_{\theta,\phi}(h_{k,u}^{D}+  \sum_{m}\hat{h}_{u,m}^{D}(\theta,\phi)g_{m}^D(\theta,\phi,\theta_1,\phi_1)\tilde{h}_{m,k}^{D}(\theta_1,\phi_1)),
    \end{aligned}
    \label{appB:counter}
\end{equation}
which indicates that $(\theta_0, \phi_0)$ is the direction of the incident signal that maximizes the power of the receiving signal at the direction of $(\theta_1, \phi_1)$. 

Consider the case where $\theta_0 \rightarrow 90\degree$, we drive the main beam direction $(\theta_1^*, \phi_1^*)$ by solving (\ref{appB:downbeam}). 
Based on (\ref{appB:counter}), the signal power of the beam at the direction $(\theta_1^*, \phi_1^*)$ can only be maximized when the incident angle is $\theta_0 \rightarrow 90\degree$. 
However, we have $cos(\theta_0) \rightarrow 0$ when $\theta_0 \rightarrow 90\degree$. 
Note that the radiation pattern $F(\theta_{i,m}^{D},\phi_{i,m}^{D})$ is proportional to $cos^{n}(\theta_0)$, were $n$ is a constant parameter \cite{tang}. 
According to (\ref{impact}), the EM response of the IOS element $g_m^{D}$ is proportional to $F(\theta_{i,m}^{D},\phi_{i,m}^{D})^{0.5}$, and thus, we have  
\begin{equation}
    \lim_{\theta_0 \rightarrow 90\degree}|g_m^{D}|=0\Rightarrow\lim_{\theta_0 \rightarrow 90\degree}|H_{k,u}^{D}|=|h_{k,u}^{D}|,
\end{equation}
which implies that only the direct link contributes to the receiving signal power and the IOS-based link has zero signal strength. Note that the power gain of the direct link is the same despite the change of the incident angle since the distance between the BS and the user is constant. Hence, there must exist another incident angle where power gain benefits from both the direct link and the IOS-assisted link. So (\ref{appB:counter}) cannot be satisfied, proving the beam non-reciprocity of IOS.
\end{proof}

\subsection{Discussion on IOS-based Channel Estimation}
The angle dependency of IOS brings a new challenge to channel estimation, especially for the high-mobility scenario.
Since the angle-domain information of the user-IOS path is time-variant due to the user mobility, the angle-dependent EM response $g_{m}^{\mathcal{X}}$ also varies with time and the accurate value of $g_{m}^{\mathcal{X}}$ is non-trivial to obtain. Nevertheless, the traditional IOS channel estimation method targets at minimizing the mean square error of the estimation result by optimizing the phase shift of each IOS element, which requires explicit information of $g_{m}^{\mathcal{X}}$. Hence, the traditional method does not fit anymore. 

%, which is further demonstrated in section IV-B.

\section{Simulation and Experimental results}
In order to validate the angle dependency of the IOS element's EM response to the signal, we perform a full-wave simulation with CST Microwave Studio and conduct experiments on an IOS. The element of the considered IOS has two mirror-symmetric layers, each of which contains a metallic patch and a pin node. The IOS element has two working states. For the ON state, both pin nodes are turned on and vice versa. The IOS works at $3.6$ GHz and the PIN diode model is BAR 65-02L. 
%The element size is set as $d_x=7.1mm$, $d_y=28.7mm$ and $d_z=14.2mm$, and the size of the transmissive/reflective patch is set as $10mm$ and $16mm$ along the x axis and the y axis, respectively. 
More details of the adopted IOS can be found in \cite{IOS}. 

\subsection{Refraction and Reflection Coefficient}
The refraction and reflection coefficient of the IOS element is simulated with CST Microwave Studio. We set a plane wave projecting toward the IOS element with a fixed azimuth angle of $0\degree$ and varied elevation angle $\theta$.
\begin{comment}
\begin{table}[t]
\centering
\caption{Amplitude of Refraction and Reflection Coefficient of IOS element when incident angle.}
\label{table1}
\begin{tabular}{|c|c|c|c|c|c|c|}
\hline
\multicolumn{2}{|c|}{Elevation angle $\theta$} & -20\degree& -10\degree &0\degree& 10\degree& 20\degree\\
\hline
\multirow{2}{*}{Reflection}& On & 0.567& 0.556 &0.538& 0.556& 0.567\\
\cline{2-7}
\multirow{2}{*}{} &Off	& 0.536& 0.486& 0.485 & 0.486 & 0.536\\ 
\hline
\multirow{2}{*}{Refraction} & On &0.794	&0.81	&0.825	&0.81	&0.794\\
\cline{2-7}
\multirow{2}{*}{} & Off &0.493	&0.551	&0.56	&0.551&0.493\\
\hline
\end{tabular}
\label{amp}
\end{table}
\end{comment}

\begin{table}[t]
\centering
\caption{Value of Phase Shift $\psi_m$ given different incident angle.}
\label{table2}
\begin{tabular}{|c|c|c|c|c|c|c|}
\hline
\multicolumn{2}{|c|}{Elevation angle $\theta$} & -20\degree& -10\degree &0\degree& 10\degree& 20\degree\\
\hline
\multirow{2}{*}{Reflection}& On & -105\degree& -135\degree& -146\degree & -135\degree & -105\degree \\
\cline{2-7}
\multirow{2}{*}{} &Off	&11\degree	&-12\degree	&-20\degree	&-12\degree	&11\degree\\ 
\hline
\multirow{2}{*}{Refraction} & On & 162\degree & 133\degree &122\degree & 133\degree & 162\degree\\
\cline{2-7}
\multirow{2}{*}{} & Off &-32\degree	&-53\degree	&-62\degree	&-53\degree&-32\degree\\
\hline
\end{tabular}
\vspace{-12pt}
\label{pha}
\end{table}

As shown in Table~\ref{pha}, phase shift $\psi_m$ varies with the incident angle, which supports Remark~\ref{remarkbeam}. As the absolute value of $\theta$ increases, the phase shift of reflected and refracted signals increases for both states of IOS. This phenomenon is contrary to the ideal assumption that $\psi_m$ is a constant value despite the change of incident angle. Since the simulated IOS has a symmetrical structure, transmission/reflection coefficients are identical for two incident signals whose elevation angles have the same absolute values.
\begin{figure}[t]
\centerline{\includegraphics[width=9cm,height=10cm]{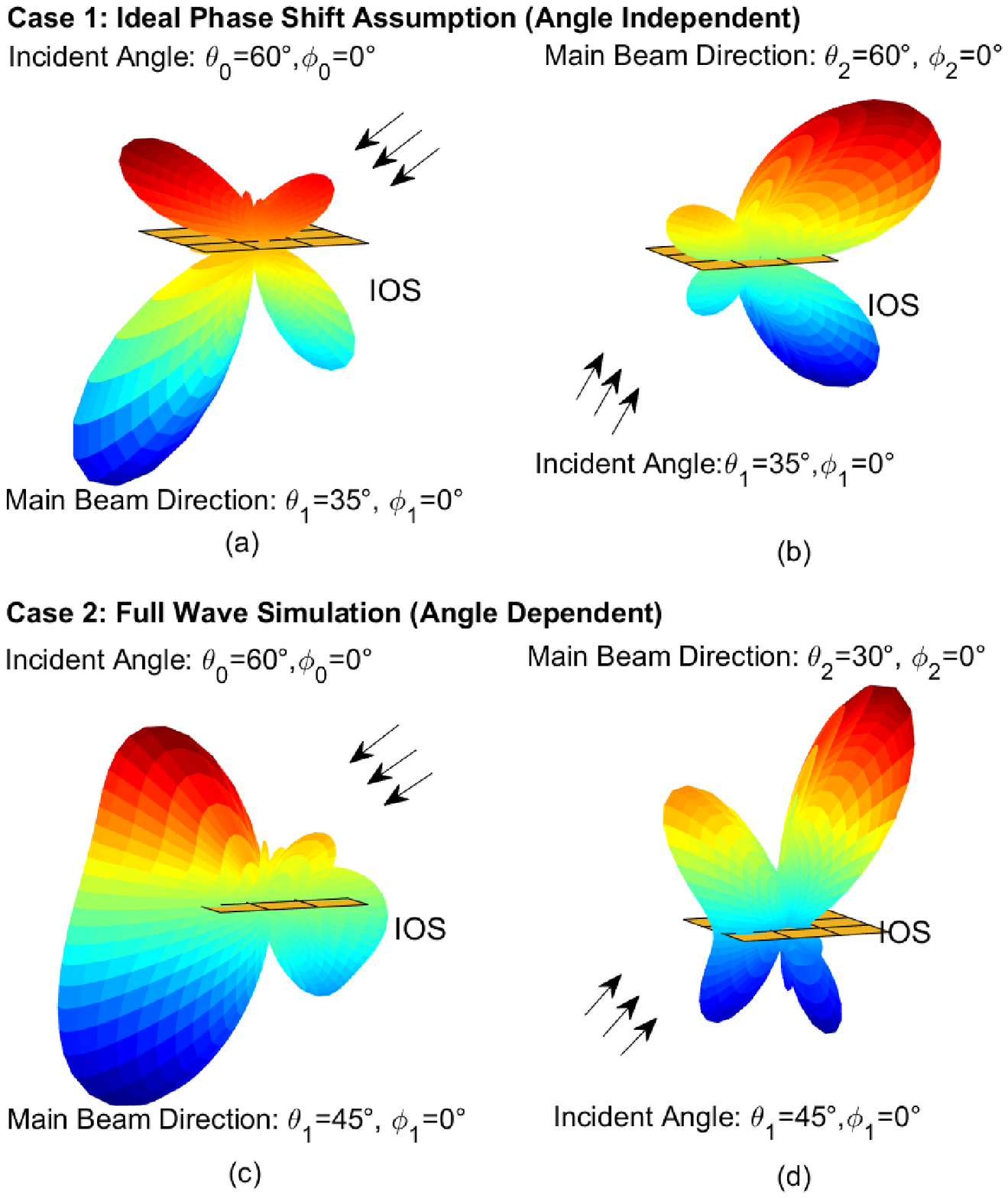}}
\caption{Beam reciprocity illustration for simulation under ideal phase assumption and full wave simulation.}
\vspace{-15pt}
\label{beam60454530}
\end{figure}

\subsection{Beam Non-reciprocity of IOS}
Fig.\ref{beam60454530} shows the radiation pattern simulation results of a $3\times3$ elements IOS with different incident angles, verifying Proposition \ref{beamremark}. To evaluate the impact of angle dependency on IOS beam reciprocity, we consider two cases for comparison: Case 1 is based on the ideal phase shift assumption and Case 2 is performed with CST Microwave Studio. 

We note that the beam reciprocity only holds in Case 1 where the ideal angle-independent assumption is adopted. The incident angle, ($\theta_0 = 60\degree,\phi_0=0\degree$), and the main beam direction, ($\theta_1 = 35\degree,\phi_1=0\degree$) in Fig.\ref{beam60454530} (a), are the main beam direction and the incident angle in Fig.\ref{beam60454530} (b), respectively. Though the results deviate from the practice, Case 1 serves as a benchmark to reflect the importance of the angle-dependent condition.  

In Fig.~\ref{beam60454530} (c) and (d), we consider the influence of angle dependency in Case 2, which verifies that the beam reciprocity does not hold in practice. In Fig.\ref{beam60454530} (c), the incident angle is ($\theta_0 = 60\degree,\phi_0=0\degree$) and the main beam direction of the refracted signal is ($\theta_1 = 45\degree,\phi_1=0\degree$). In Fig.\ref{beam60454530} (d), the incident angle is $(\theta_1 = 45\degree,\phi_1=0\degree)$ and the main beam direction of the refracted signal is $(\theta_2 = 30\degree,\phi_2=0\degree)$ with the same IOS configuration in Fig.~\ref{beam60454530} (c). Since $(\theta_0, \phi_0)\neq(\theta_2,\phi_2)$, a counter example of beam reciprocity is found, which proves the beam non-reciprocity characteristic of IOS. 

\begin{figure}[t]
%\centering
\begin{minipage}[t]{1\linewidth}
\includegraphics[width=8cm,height=6cm]{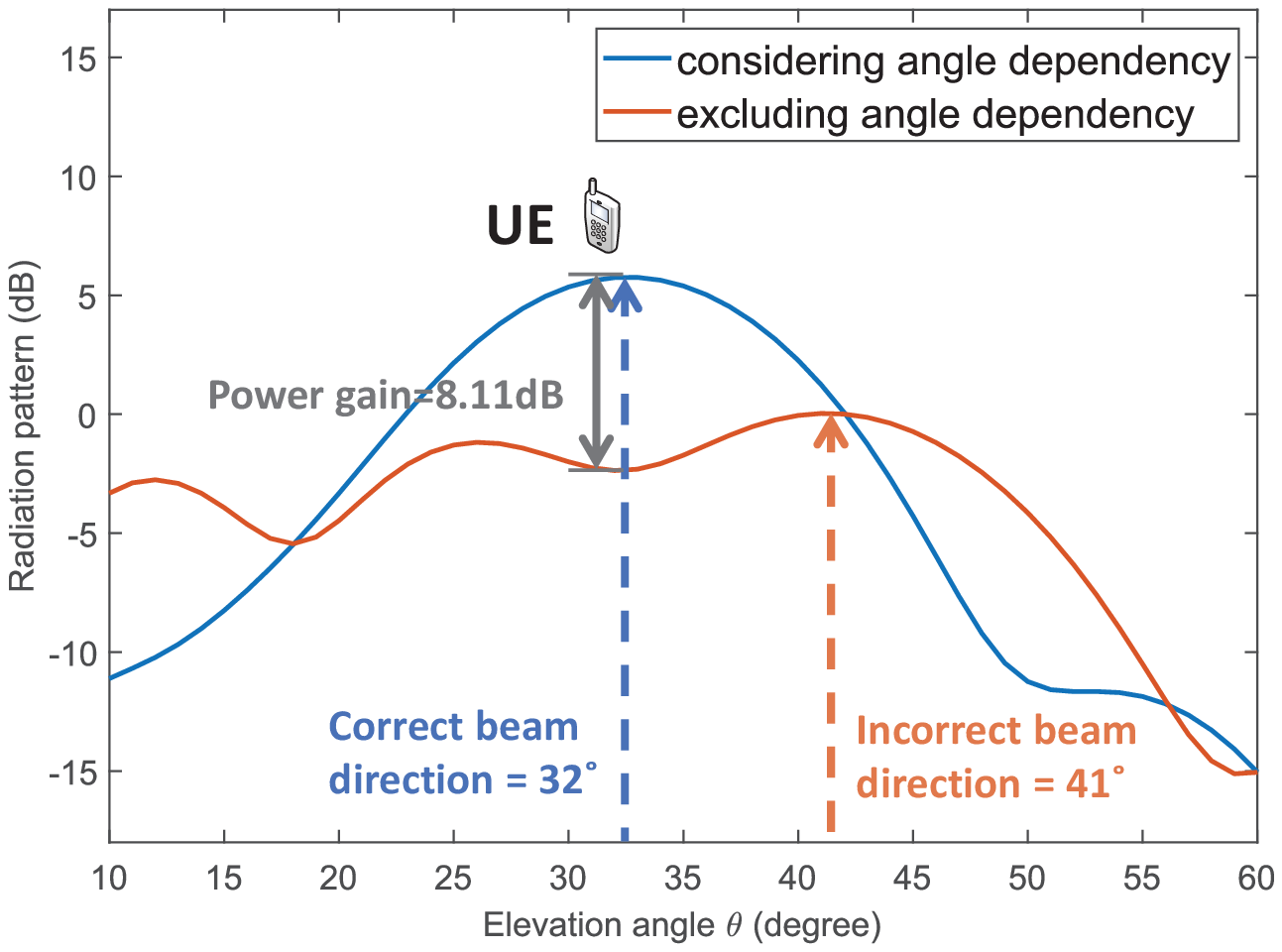}
\caption{IOS-assisted beamforming at the refractive zone.}
\label{tran_beamform}
\includegraphics[width=8cm,height=6cm]{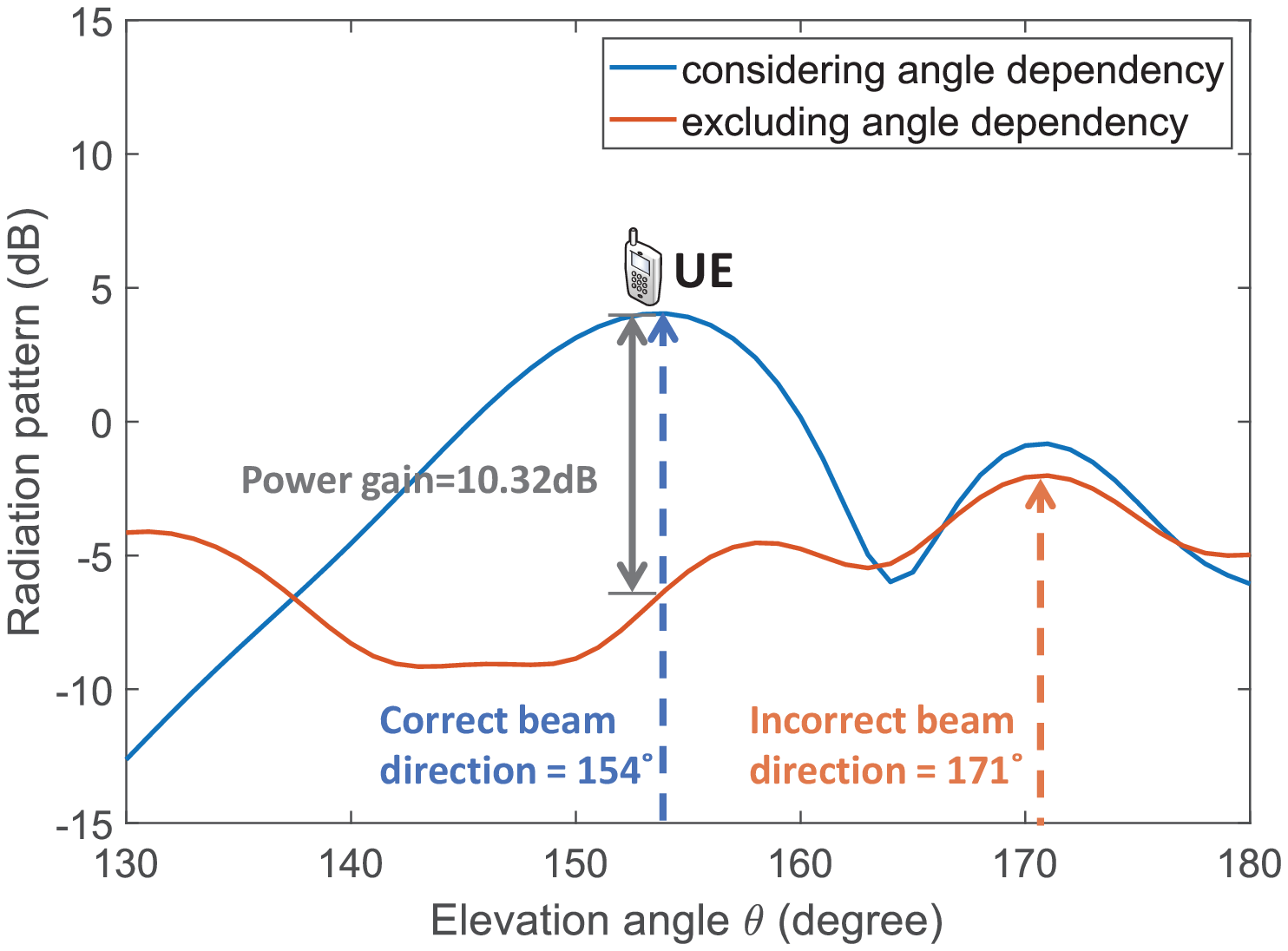}
\caption{IOS-assisted beamforming at the reflective zone.}
\vspace{-10pt}
\label{refl_beamform}
\end{minipage}
\end{figure}

\subsection{Influence of the IOS Angle Dependency on Beamforming}
Fig.~\ref{tran_beamform} and Fig.~\ref{refl_beamform} illustrate the beamforming under the conditions that the angle dependency of IOS is considered or excluded for comparison to illustrate Remark~\ref{remarkbeam}. 
Conventional IOS-assisted beamforming method optimizes array configuration with the ideally assumed phase shift $\psi_m \in \{0, \pi\}$ for each element regardless of the incident and departure angle. However, as shown in TABLE ~\ref{pha}, $\psi_m$ has diverse possible values with different incident angle. The ideal phase shift assumption can lead to performance degradation for IOS beamforming.  
In the simulation, we consider an IOS array with $42 \times 42$ elements located at the x-y plane. The transmitter is placed at $(0m,0m,0.15m)$. The directions of users are set as $(\theta=32\degree, \phi = 0\degree)$ at the refractive zone and $(\theta=154\degree, \phi=0\degree)$ at the reflective zone.   

 As shown in Fig.~\ref{tran_beamform} and Fig.~\ref{refl_beamform}, considering the angle dependent phase shift of the IOS element, beams can correctly point at the users at the refractive zone and the reflective zone. On the contrary, if the angle dependency is excluded from beamforming, beams will have $9\degree$ and $17\degree$ of error at the refraction zone and the reflective zone, respectively. Besides, the radiation power gains at the directions of the users diminish for $8.11$ dB and $10.32$ dB for the refraction case and the reflection case, respectively. Hence, the angle dependency of $\psi_m$ should be considered for correct IOS-assisted beamforming.   
 \vspace{-10pt}

% \centerline{\includegraphics[width=8cm,height=6cm]{transmit.eps}}
% \caption{IOS-assisted beamforming at the refractive zone.}
% \label{tran_beamform}
% \end{figure}

% \begin{figure}[htbp]
% \centerline{\includegraphics[width=8cm,height=6cm]{reflect.eps}}
% \caption{IOS-assisted beamforming at the reflective zone.}
% \label{refl_beamform}
% \end{figure}

\subsection{Experiment on IOS Channel Reciprocity}
To verify Proposition~\ref{proposition1} and Proposition~\ref{proposition2}, we conduct an experiment on IOS channel reciprocity. As shown in Fig.~\ref{ch_measure}, we set up a channel reciprocity measurement testbed for the IOS-assisted communication system in a microwave anechoic chamber deployed in \cite{IOS}. The adopted IOS consists of 320 reconfigurable elements and is controlled by a field-programmable gate array (FPGA), to which a computer sends array configuration command. The inside of the chamber is covered with wave absorbing material to simulate the free space. Two horn antennas are connected to the two ports of the vector network analyzer to measure the scattering 
parameters of the IOS-aided communication system.

\begin{figure}[t]
\centerline{\includegraphics[width=9.2cm,height=5.8cm]{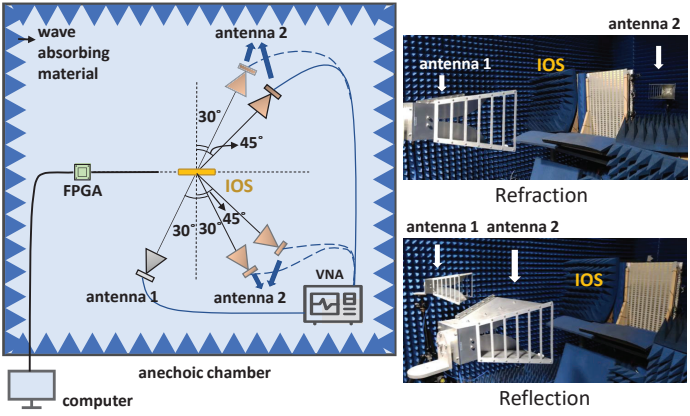}}
\caption{IOS channel reciprocity measurement environment.}
\vspace{-10pt}
\label{ch_measure}
\end{figure}

The antennas point at the center of the IOS array. The direction of antenna 1 to IOS is set as $\theta_{r/i,m}^{\mathcal{X}} = 30\degree$, and the direction of antenna 2 to IOS is set as $\theta_{i/r,m}^{\mathcal{X}} = 30\degree$ and $\theta_{i/r,m}^{\mathcal{X}} = 45\degree$ at both the reflective zone and the refractive zone of the IOS. For both antennas, we set $\phi_{i/r,m}^{\mathcal{X}}=0\degree$. The uplink channel, from antenna 1 to antenna 2, is represented by the S21 parameter measured by the vector network analyzer. The downlink channel, from antenna 2 to antenna 1, is represented by the S12 parameter. 

Fig.~\ref{reflect} shows the measurement of the uplink channel and the downlink channel where antenna 2 is placed at the reflective zone and the refractive zone of the IOS, respectively. For the refraction case, IOS configuration 1 maximizes the signal power at the direction of $\theta_{i/r,m}^{\mathcal{X}}=30\degree$; configuration 2 maximizes the signal power at the directon of $\theta_{i/r,m}^{\mathcal{X}}=45\degree$. For the reflection case, configuration 3 maximizes the signal power at the direction of $\theta_{i/r,m}^{\mathcal{X}}=30\degree$; configuration 4 maximizes the signal power at the directon of $\theta_{i/r,m}^{\mathcal{X}}=45\degree$. The uplink channel equals the downlink channel when the incident signal is either refracted or reflected, given different IOS configurations and signal directions, which verifies Proposition~\ref{proposition2}.
\vspace{-12pt}

\begin{figure}[htbp]
\centerline{\includegraphics[width=8.4cm,height=6.4cm]{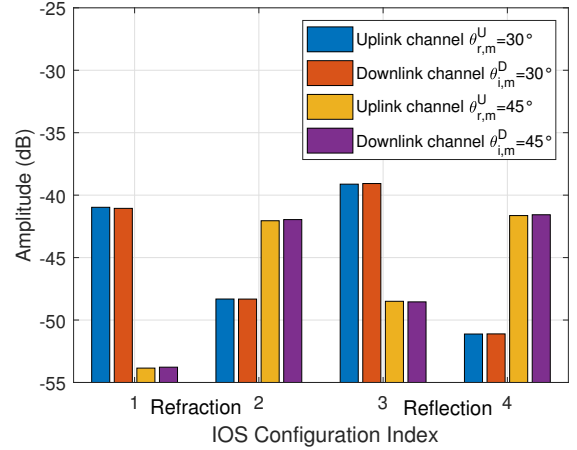}}
\caption{Channel measurement for signal refraction and reflection.}
\vspace{-0.2cm}
\label{reflect}
\end{figure}

%\begin{figure}[htbp]
%\centerline{\includegraphics[width=8cm,height=6cm]{transmit2.eps}}
%\caption{Channel measurement when antenna 2 is placed at the refractive zone.}
%\label{transmit}
%\end{figure}

\section{Conclusion}
In this letter, we have considered an IOS-assisted
communication system with one BS and multiple users. We derived the IOS-based channel model, based on which the channel reciprocity and beam reciprocity are analyzed. We conducted simulation and experiments on IOS and draw the following conclusions. \emph{First}, the IOS-based uplink channel and downlink channel are proved to be equal. \emph{Second}, the EM response of IOS element to the signal is affected by both the angles of incidence and departure. Such an angle-dependent property should be considered for IOS-based beamforming to improve accuracy. \emph{Third}, beam reciprocity does not hold for IOS. Thus, even when the transmitter and the receiver remain static, IOS should be configured independently for accurate uplink and downlink beamforming, respectively.

\begin{comment}
\begin{appendices}

\section{Proof of Channel Model $h_{k,m}^{d}$ and $h_{m,s}^{d}$}\label{app_channelmodel}
According to the Friis Equation, the power received by the $m^{th}$ IOS element from the $k^{th}$ antenna of the BS is denoted as
\begin{equation}
    P_{m}^{r}=\frac{G_{k}F_{k}(\theta_{i,m},\phi_{i,m})A_{m}(\theta_{i,m},\phi_{i,m})}{4\pi (d_{U-B})^{2}}P^{t}_{k},
\end{equation}
where $A_m$ is the effective receiving area of the IOS element, whose influence on the channel is considered in $g_{m}^{\mathcal{X}}$. Therefore,
\begin{equation}
    h_{k,m}^{d}=\sqrt{\frac{P_{m}^{r}}{A_{m}(\theta_{i,m},\phi_{i,m})P^{t}_{k}}}e^{-j\frac{2\pi d_{k,m}}{\lambda}},
\end{equation}
which proves (\ref{b2i}). The power received by the user antenna from the $m^{th}$ IOS element is denoted as
\begin{equation}
    P_{u}^{r}=\frac{\lambda^{2}GF(\theta_{r,m},\phi_{r,m})G_{U}F_{u}(\theta_{r,m},\phi_{r,m})}{(4\pi d_{u,m})^{2}}P^{t}_{m}.
\end{equation}
Since the influence of the IOS element is considered in $g_{m}^{\mathcal{X}}$, 
\begin{equation}
    h_{m,u}^{d} = \sqrt{\frac{P_{U}^{r}}{GF(\theta_{r,m},\phi_{r,m})P^{t}_{I,m}}}e^{-j\frac{2\pi d_{u,m}}{\lambda}}
\end{equation}
which proves (\ref{i2u}). 

\section{Proof of Beam Non-reciprocity of IOS}\label{beamreciprove}
\end{appendices}
\end{comment}
%\vspace{-0.1cm}
%%%%%%%%%%%%%%%%%%%%%%%%%%%%%%%%%%%%%%%%%%%%%%%


\begin{thebibliography}{0}
\bibitem{RISintro}
M. Di Renzo \emph{et al.}, “Smart Radio Environments Empowered by Reconfigurable Intelligent Surfaces: How It Works, State of Research, and The Road Ahead,” \emph{IEEE J. Sel. Areas Commun.}, vol. 38, no. 11, pp. 2450-2525, Nov. 2020.

\bibitem{IOSintro}
H. Zhang \emph{et al.}, “Intelligent omni-surfaces for full-dimensional wireless communications: Principles, technology, and implementation,”
\emph{IEEE Commun. Mag.}, vol. 60, no. 2, pp. 39–45, Feb. 2022.

\bibitem{IOS}
 S. Zeng \emph{et al.}, “Intelligent omni-surfaces:
Reflection-refraction circuit model, full-dimensional beamforming, and system implementation,” [Online]. Available:
https://arxiv.org/abs/2206.00204.

\bibitem{IOScover}
S. Zeng \emph{et al.}, "Reconfigurable Intelligent Surfaces in 6G: Reflective, Transmissive, or Both?," \emph{IEEE Commun. Lett.}, vol. 25, no. 6, pp. 2063-2067, Jun. 2021.

\bibitem{IOS_impact}
S. Zhang \emph{et al.}, “Intelligent Omni-Surfaces: Ubiquitous Wireless Transmission by Reflective-Refractive Metasurfaces,” \emph{IEEE Trans. Wireless Commun.}, vol.~21, no.~1, pp.~219-233, Jan.~2022.

%\bibitem{IOS_beamforming}
%W. Cai \emph{et al.}, "Joint Beamforming Designs for Intelligent Omni Surface Assisted Wireless Communication Systems," in \emph{Proc. IEEE Glob. Commun. Conf. (GLOBECOM)}, Madrid, Spain, Dec. 2021, pp. 1-6.

%\bibitem{poweropt}
%Y. Wang, P. Guan, H. Yu and Y. Zhao, "Transmit Power Optimization of Simultaneous Transmission and Reflection RIS Assisted Full-Duplex Communications," \emph{IEEE Access}, vol. 10, pp. 61192-61200, May, 2022.

%\bibitem{sumrateopt}
%H. Niu, \emph{et al.}, "Weighted Sum Rate Optimization for STAR-RIS-Assisted MIMO System," \emph{IEEE Trans. Veh. Technol.}, vol. 71, no. 2, pp. 2122-2127, Feb. 2022.

%\bibitem{outagepana}
%C. Zhang, \emph{et al.}, "Simultaneously Transmitting And Reflecting RIS Aided NOMA With Randomly Deployed Users," in \emph{Proc. IEEE Glob. Commun. Conf. (GLOBECOM)}, Madrid, Spain, Dec. 2021, pp. 1-6.

\bibitem{IOSana}
J. Xu, Y. Liu and X. Mu, “Performance Analysis for the Coupled Phase-Shift STAR-RISs,” in \emph{IEEE Wireless Commun. Networking Conf. (WCNC)}, Austin, Texas, USA, Apr. 2022, pp. 489-493.

% \bibitem{intro}
% S. Zeng \emph{et al.}, "Reconfigurable Intelligent Surfaces in 6G: Reflective, Transmissive, or Both?," \emph{IEEE Commun. Lett.}, vol.~25, no.~6, pp.~2063-2067, June 2021.

\bibitem{tang}
W. Tang \emph{et al.}, “Wireless Communications With Reconfigurable Intelligent Surface: Path Loss Modeling and Experimental Measurement,” \emph{IEEE Trans. Wireless Commun.}, vol.~20, no.~1, pp.~421-439, Jan. 2021. 



\bibitem{star_ches}
C. Wu, C. You, Y. Liu, X. Gu and Y. Cai, “Channel Estimation for STAR-RIS-aided Wireless Communication,” \emph{IEEE Commun. Lett.}, vol.~26, no.~3, pp. 652-656, Mar. 2022.

\bibitem{gmreciprocity}
V. S. Asadchy et al., “Tutorial on Electromagnetic Nonreciprocity and Its Origins,” \emph{Proc. IEEE}, vol. 108, no. 10, pp. 1684–1727, Oct. 2020.

\bibitem{radiationpatternref}
Constantine A. Balanis, “Fundamental Parameters and Figures-of-Merit of Antennas”, in \emph{Antenna Theory Analysis and Design}, 4th ed. Hoboken: John Wiley \& Sons, Inc., 2016.

\bibitem{idealphase}
H. Gao, K. Cui, C. Huang and C. Yuen, “Robust Beamforming for RIS-Assisted Wireless Communications With Discrete Phase Shifts,” \emph{IEEE Wireless Commun. Lett.}, vol. 10, no. 12, pp. 2619-2623, Dec. 2021.
\end{thebibliography}
\end{document}